\documentclass[journal]{IEEEtran}
\usepackage[latin1]{inputenc}
\usepackage{amsmath}
\usepackage{graphicx}
\usepackage{amssymb}
\usepackage{xcolor}
\makeatletter

\newtheorem{lemma}{Lemma}

\newtheorem{thm}{Theorem}

\providecommand{\LyX}{L\kern-.1667em\lower.25em\hbox{Y}\kern-.125emX\@}

\def\BibTeX{{\rm B\kern-.05em{\sc i\kern-.025em b}\kern-.08em
    T\kern-.1667em\lower.7ex\hbox{E}\kern-.125emX}}

\makeatother
\begin{document}
\def\ZZ{{\mathbb Z}}
\def\RR{{\Bbb R}}
\def\NN{{\mathbb N}}
\def\CC{{\mathbb C}}

\title{Observer-Based Target Control for Mismatched Time-Delay Systems}
	\author{Hieu Trinh
	%\IEEEmembership{Senior Member, IEEE}
	%\thanks{This paper was submitted to IEEE Transactions on Automatic Control for review on 15 January 2026.}
	\thanks{Hieu Trinh is with the School of Engineering, Deakin University, Waurn Ponds, 75 Pigdons Road, Geelong, Australia. (email:  hieu.trinh@deakin.edu.au)}
	}

\maketitle \maketitle \maketitle \thispagestyle{plain}
\pagestyle{plain}

\begin{abstract}
This paper addresses observer-based target control for linear time-delay systems subject to simultaneous, mismatched input and output latencies. While full-state regulation is often conservative and computationally intensive, practical engineering objectives typically require controlling only specific linear combinations of states, or target outputs. To overcome the challenges posed by these asymmetric, dual-channel delays, we propose a reduced-order modeling framework inspired by the structural philosophy of Fernando and Darouach \cite{Fernando2025}. By projecting the high-dimensional plant dynamics onto the row space of the target output matrix $F_o$, the controller focuses strictly on the lower-dimensional target subspace. Based on this projection, an observer-based control scheme is developed to ensure precise target stabilization despite the simultaneous, mismatched input, state, and output latencies.

\end{abstract}


\section{System Description and Problem Statement}
Consider the following time-delay system:
\begin{align}
	\label{d1}
	\dot{x}(t)&=Ax(t)+A_dx(t-\tau_x)+Bu(t-\tau_u),\\
	\label{d2}
	y(t)& = Cx(t-\tau_y),
\end{align}
where $x(t)\in \mathbb{R}^n$ is the state vector, $u(t)\in \mathbb{R}^r$ is the control input vector, and $y(t)\in\mathbb{R}^{p}$ is the measured output vector. The constants $\tau_x>0$, $\tau_u>0$ and $\tau_y>0$ represent the time delays in the state, control input, and output channels, respectively. The initial condition is defined by the function $x(t) = \rho(t)$ for $t \in [-\tau_{\max}, 0]$, where $\tau_{\max} = \max\{\tau_x, \tau_y\}$. The system matrices $A,A_d\in\mathbb{R}^{n\times n}$, $B\in\mathbb{R}^{n\times r}$, and $C\in\mathbb{R}^{p\times n}$ are constant. Without loss of generality, it is assumed that $B$ has full column rank and $C$ has full row rank. We do not assume that $A$ is a Hurwitz matrix. In fact, even if $A$ is Hurwitz, the overall system \eqref{d1} may still be destabilized by the presence of the delayed term $A_d x(t-\tau_x)$. 

The target output vector to be regulated is defined as:
\begin{equation}\label{d3}
	z_o(t)=F_ox(t),
\end{equation}
where $z_o(t)\in \mathbb{R}^{m}$ and $F_o\in \mathbb{R}^{m\times n}$ is a constant matrix. We assume $F_o$ has full row rank ($m < n$) because the linear functions targeted for control are linearly independent; ensuring the regulation of these independent components inherently guarantees the control of any linearly dependent variations.

The primary objective of this paper is to design an observer-based controller that asymptotically drives the target output vector to the origin, i.e., $\lim_{t \to \infty} z_o(t) = \mathbf{0}$, from any initial condition.

A key contribution of this work is the development of a reduced-order observer-based design, achieved by projecting the full state dynamics directly onto the row space of the target output matrix $F_o$ \cite{Fernando2025}. While traditional approaches burden the observer with reconstructing the entire high-dimensional state vector $x(t)$, the proposed projection method strategically isolates the lower-dimensional subspace that actively dictates the target tracking profile. Specifically, to stabilize the target output under severe latencies, we introduce a delay-compensated control strategy that regulates the linear combination $F_ox(t)$ rather than the full state vector. As noted in \cite{Fernando2025}, the asymptotic stability of the target output does not inherently require the asymptotic stability of the entire state vector $x(t)$. This distinction highlights the practical advantage of targeting a specific subspace, avoiding the unnecessary conservative constraints of full-state stabilization.

To handle the time delay in the input channel effectively and render the synthesis problem tractable, we consider the case where $\tau_x>\tau_u$ (the case where $\tau_u > \tau_x$ is discussed in Remark 2). The proposed control law is structured as follows:
\begin{align}
	\label{d4}u(t) = Z_{\tau_u}z_o(t) + Z_{\tau_x}z_o(t - \tau_x + \tau_u), 
\end{align}
which, when subject to the input delay $\tau_u$, yields the delayed control input:
\begin{align}
	\label{d5}
	u(t-\tau_u) = Z_{\tau_u}z_o(t-\tau_u) + Z_{\tau_x}z_o(t-\tau_x) \end{align}
where $Z_{\tau_u}, Z_{\tau_x} \in \mathbb{R}^{r \times m}$ are the control gains to be designed such that the target output vector $z_o(t)\to \bf 0$ asymptotically. Note that when $\tau_u = \tau_x$, the control law naturally reduces to $u(t-\tau_u) = Z_{\tau_u}z_o(t-\tau_u)$.

Importantly, the control law is applied consistently to both the full-order and reduced-order models. To ensure a valid framework, the spectral properties of the projected closed-loop system must be preserved within those of the full-order system. This guarantees that any stability or instability observed in the reduced-order system directly corresponds to that of the full system. We will establish this relationship in the subsequent sections of this paper.

In practice, however, neither $z_o(t)$ nor its delayed counterpart $z_o(t-\tau_u)$ is directly available for feedback. Furthermore, the measured output vector is subject to the time delay modeled in (\ref{d2}). To overcome these limitations, the designed controller is implemented using a functional observer, the details of which will be discussed in the subsequent sections of this paper.

The remainder of this paper is organized as follows. Section II presents the essential lemmas required to derive the reduced-order time-delay system, which is obtained by projecting the full-order state dynamics onto the row space of the target output matrix $F_o$. By capturing the evolution of the target output variables, this projected system serves as a lower-dimensional representation for assessing stabilizability. Consequently, the target output controller is designed by stabilizing this lower-order subsystem. Section III details the design of a functional dual-observer architecture to estimate the delayed control law (\ref{d5}). This scheme extends the framework developed in \cite{trinh5} to the class of time-delay systems (\ref{d1})-(\ref{d2}) characterized by mismatched latencies.

\section{Target Output Control via Lower-Order Time-Delay Subsystems}

We will use the following lemmas \cite{Fernando2025} in the sequel.
\begin{lemma}[\cite{Fernando2025}]\label{lemma1p5}
	The following statements are equivalent:
	\begin{itemize}
		\item[i)] $\mathrm{rank}\begin{pmatrix}
			F_oA\\F_o
		\end{pmatrix} =\mathrm{rank}(F_o)$,
		\item[ii)] $	F_oA(I-F^{-}_oF_o) = \mathbf{0}.$  
	\end{itemize}
\end{lemma}
\begin{lemma}[\cite{Fernando2025}]\label{lemma2p5}
	The following equation
	\begin{equation} 
		N_oF_o-F_oA=\bf 0, \nonumber 
	\end{equation} 
	where $A\in \mathbb{R}^{n\times n}$, $F_o\in \mathbb{R}^{m\times n}$, $\mathrm{rank}(F_o)=m$ and $m \leq n$ are known matrices and 
	$N_o \in \mathbb{R}^{m\times m}$ is an unknown matrix, has a solution
	if and only if
	\begin{equation} 
		\mathrm{rank}\begin{pmatrix}
			F_oA\\F_o
		\end{pmatrix} =\mathrm{rank}(F_o), \nonumber 
	\end{equation}
	and in this situation $N_o=F_oAF^-_o$ satisfies
	\begin{align}
		\sigma(N_o) \subseteq \sigma(A), \nonumber 
	\end{align}
where $\sigma(N_o)$ denotes the spectrum (the set of all eigenvalues) of the matrix $N_o$.\end{lemma}

Based on Lemma \ref{lemma1p5}, a reduced-order system is derived by projecting the full-order state dynamics onto the row space of the target output matrix $F_o$.

Pre-multiplying \eqref{d1} by $F_o$, we obtain
\begin{equation} 
	F_o\dot{x}(t) = F_oAx(t) + F_oA_dx(t -\tau_x) + F_oBu(t-\tau_u) \nonumber
\end{equation}	
which can be expressed as
\begin{align}
	\dot{z}_o(t)&= F_oA(I-F^-_oF_o+F^-_oF_o)x(t) \nonumber\\
	&+F_oA_d(I-F^-_oF_o+F^-_oF_o)x(t-\tau_x) +F_oBu(t-\tau_u) \nonumber\\
	&=N_oz_o(t) + N_{od}z_o(t-\tau_x) +B_ou(t-\tau_u)\nonumber\\&+F_oA(I-F^-_oF_o)x(t)+F_oA_d(I-F^-_oF_o)x(t-\tau_x), \label{d6}
\end{align}
where $N_o:=F_oAF^-_o$, $N_{od}:=F_oA_dF^-_o$, $B_o:=F_oB$, and $F_o^-$ is a generalized inverse of the matrix $F_o$ satisfying the condition $F_o F_o^- F_o = F_o$.

From Lemma \ref{lemma1p5}, the following condition
\begin{align}
	\label{d7}
	&\mathrm{rank}\begin{pmatrix}
		F_oA\\F_o
	\end{pmatrix} =\mathrm{rank}\begin{pmatrix}
	F_oA_d\\F_o
\end{pmatrix}=\mathrm{rank}(F_o)
\end{align}
is equivalent to the following two equations:
\begin{align}
	F_oA(I - F_o^-F_o)&= \mathbf{0}, \label{d8} \\
	F_oA_d(I - F_o^-F_o) &= \mathbf{0}. \label{d9}
\end{align}
Therefore, if \eqref{d7} holds, then from the projected dynamics in \eqref{d6}, we obtain the following reduced-order time-delay subsystem of dimension $m$:
\begin{align}
	\dot{z}_o(t)&=N_oz_o(t) + N_{od}z_o(t-\tau_x) +B_ou(t-\tau_u). \label{d10}
\end{align}

Substituting the control law (\ref{d5}) into (\ref{d10}) yields the closed-loop time-delay system
\begin{align}\label{d11}\dot{z}_o(t)= N_oz_o(t) +N_{\tau_u}z_o(t-\tau_u)+N_{\tau_x}z_o(t-\tau_x),\end{align}
where $N_{\tau_u}=B_oZ_{\tau_u}$ and $N_{\tau_x}=N_{od}+B_oZ_{\tau_x}$. The closed-loop time-delay system (\ref{d11}) is asymptotically stable if all complex roots $\lambda$ satisfying the characteristic equation
\begin{align}\label{d12}\det\Big(\lambda I - N_o - N_{\tau_u}e^{-\lambda \tau_u}- N_{\tau_x}e^{-\lambda \tau_x}\Big) = 0\end{align}lie in the open left-half complex plane (i.e., $\text{Re}(\lambda) < 0$). Under this condition, $z_o(t) \rightarrow \mathbf{0}$ as $t \rightarrow \infty$ for any given initial function, meaning the target output is successfully stabilized via a reduced subsystem of order $m$.

Note that for given delays $\tau_u$ and $\tau_x$, Lemma 13 in \cite{trinhnam26} can be employed to determine the gain matrices $Z_{\tau_u}$ and $Z_{\tau_x}$. Specifically, a sufficient condition for the closed-loop stability of (\ref{d11}) is guaranteed by the feasibility of the linear matrix inequality (LMI) formulated in that lemma.

Thus far, Lemma \ref{lemma1p5} has been utilized to derive the reduced-order time-delay system (\ref{d10}), and the closed-loop asymptotic stability of the reduced system (\ref{d11}) has been established based on the generalized delayed control law (\ref{d5}). Because this control law is ultimately applied to the actual full-order system, it is crucial that the spectral properties of the reduced-order closed-loop system are preserved within the full-order dynamics. Ensuring this spectral inclusion guarantees that the stability properties verified on the lower-dimensional surrogate directly govern the behavior of the full system, establishing a mathematically valid framework for controller design via the reduced-order dynamics. To formally characterize this relationship, the following development establishes the exact mapping between the eigenvalue spectra of the full- and reduced-order closed-loop systems.

To simplify the analysis, we consider the scenario where $\tau_u = \tau_x$ under the following control law:
\begin{align}
	\label{d13}u(t-\tau_x)=Z_{\tau_u}z_o(t-\tau_x).\end{align}
Applying the above control law to the full-order system (\ref{d1}) yields the following closed-loop dynamics:
$$\dot{x}(t)=Ax(t)+A_dx(t-\tau_x)+BZ_{\tau_u}F_ox(t-\tau_x),$$whose spectrum is given by$$\sigma\Big(A +A_de^{-\lambda \tau_x} + BZ_{\tau_u}F_oe^{-\lambda \tau_x}\Big).$$
Similarly, by applying the control law (\ref{d13}) to the reduced-order system (\ref{d10}) yields the following closed-loop dynamics:
$$\dot{z}_o(t)=N_oz_o(t) + N_{od}z_o(t-\tau_x) +B_oZ_{\tau_u}z_o(t-\tau_x),$$whose spectrum is given by$$\sigma\Big(N_o+N_{od}e^{-\lambda \tau_x} + B_oZ_{\tau_u}e^{-\lambda \tau_x}\Big).$$

Subject to the satisfaction of the rank condition (\ref{d7}), the following matrix equations are satisfied:
\begin{align}
	\label{d14}
	N_oF_o-F_oA&=\bf 0,\\
	\label{d15} N_{od}F_o-F_{o}A_d&=\bf 0.\end{align}

Based on (\ref{d14})-({\ref{d15}), it is easy to show that the following holds:
\begin{align}
	\label{d16}
	\bar{N}_oF_o-F_o\bar{A}=\bf 0,\end{align}
where $\bar{N}_o:=N_o+N_{od}e^{-\lambda \tau_x} + B_oZ_{\tau_u}e^{-\lambda \tau_x}$ and
$\bar{A}:=A +A_de^{-\lambda \tau_x} + BZ_{\tau_u}F_oe^{-\lambda \tau_x}$.

Since $F_o\neq \bf 0$ is full row rank, from Lemma \ref{lemma2p5}, it follows that \[
\sigma\left( \bar{N}_o \right)
\subseteq
\sigma\left( \bar{A} \right).
\]

This finding establishes that the spectrum of the projected closed-loop system is preserved within that of the full-order closed-loop system, providing a valid framework for designing controllers using reduced-order dynamics. When using the control law \eqref{d5} to stabilize \eqref{d10}, similar spectral properties can be established; this derivation is omitted here for brevity.

Let us pause here for an illustrative example.

\textit{Example 1:} Consider system (\ref{d1}) with time delays $\tau_u = \tau_x = 0.5\text{s}$ and the following system matrices:

$A=\begin{pmatrix} 
	1 &0.5 &-1 &0 &1\\0.3 &0.5 &-0.6 &-0.3 &0.3\\-0.6 &0 &0.2 &0.6 &-0.6\\1.25 &0.5 &-1 &-0.25 &1.75\\-0.75 &0 &0 &0.75 &-0.25 \end{pmatrix}$,

$A_d=\begin{pmatrix} 2.5 &1.25 &-2.5 &0 &2.5\\
0.75 & 1.25 &-1.5 &-0.75 &0.75\\
-1.5 &0 &0.5 &1.5 &-1.5\\
3.125 &1.25 &-2.5 &-0.625 &4.375\\
-1.875 &0 &0 &1.875 &-0.625\end{pmatrix}$,
 
$B=\begin{pmatrix}1 &-1\\1 &1\\0 &0\\1 &0\\0 &1 \end{pmatrix}$.

Our objective is to design a controller that stabilizes the target output vector (\ref{d3}), where
$$F_o=\begin{pmatrix}0.5 &1 &-2 &0.5 &2.5\\0.75 &1 &-2 &0.25 &2.25 \end{pmatrix}.$$

Simulation easily verifies that the target output vector is unstable. To address this, we design a controller of the form (\ref{d13}) based on the reduced-order system (\ref{d10}) using the following procedure.

First, we verify the rank condition (\ref{d7}) and find that it is satisfied. Accordingly, we obtain the reduced-order time-delay system (\ref{d10}), where

$N_o:=F_oAF^-_o=\begin{pmatrix}0 &1\\-0.5 &1.5\end{pmatrix}$,  $N_{od}:=F_oA_dF^-_o=\begin{pmatrix}0 &2.5\\-1.25 &3.75\end{pmatrix}$, and $B_o:=F_oB=\begin{pmatrix}2 &3\\2 &2.5\end{pmatrix}$.

Next, we design a controller (\ref{d13}) to stabilize the reduced-order system. For the given $\tau_u=0.5\text{s}$, the LMI in Lemma 11 \cite{trinhnam26} is feasible for $\lambda=1$, and we obtain
\[Z_{\tau_u}=\begin{pmatrix} 5.2354  & -7.7006\\-3.7586 &4.0337 \end{pmatrix}.\]
Consequently, the target output controller 
\[u(t-0.5)=Z_{\tau_u}z_o(t-0.5),\]
guarantees the asymptotic stability of the following closed-loop target output vector
\begin{align}
	\dot{z}_o(t)&=\begin{pmatrix}0 &1\\-0.5 &1.5\end{pmatrix}z_o(t)\nonumber\\&+\begin{pmatrix}-0.8050  & -0.7999\\
		-0.1757  & -1.5668 \end{pmatrix}z_o(t-0.5).\nonumber\end{align}
For the above time-delay system, the eigenvalues located to the right of the vertical line $\alpha = -4$ are calculated as $\{-0.4537\pm j1.5519, -0.4646, -3.6765\}$ using the approach described in \cite{wu}.

Now, by applying the same target output controller to the full-order system (\ref{d1}), we obtain the closed-loop system
\begin{align}
	\label{d17}
	\dot{x}(t)=Ax(t)+(A_d+BZ_{\tau_u}F_o)x(t-0.5).\end{align}
The eigenvalues of the above closed-loop system satisfying $\operatorname{Re}(s) > -4$ are evaluated as
\\
\\
$\{1.1896,\, 0.5751,\, -0.6112 \pm j3.3712\} \cup \{-0.4537 \pm j1.5519,\, -0.4646,\, -3.6765\}$. 
\\
\\
This resulting set contains the four stable eigenvalues belonging to the reduced-order closed-loop system.

Figure~\ref{fig1paper6} displays the trajectories of $z_{o1}(t)$ and $z_{o2}(t)$, which clearly converge to $0$ as $t \rightarrow \infty$. Although the closed-loop system (\ref{d17}) is inherently unstable, the target output vector $z_o(t)$ is successfully stabilized by the delayed target output controller $u(t-0.5)=Z_{\tau_u}z_o(t-0.5)$.
\begin{figure}[!h]
	\centering
	\includegraphics[width=0.9\linewidth]{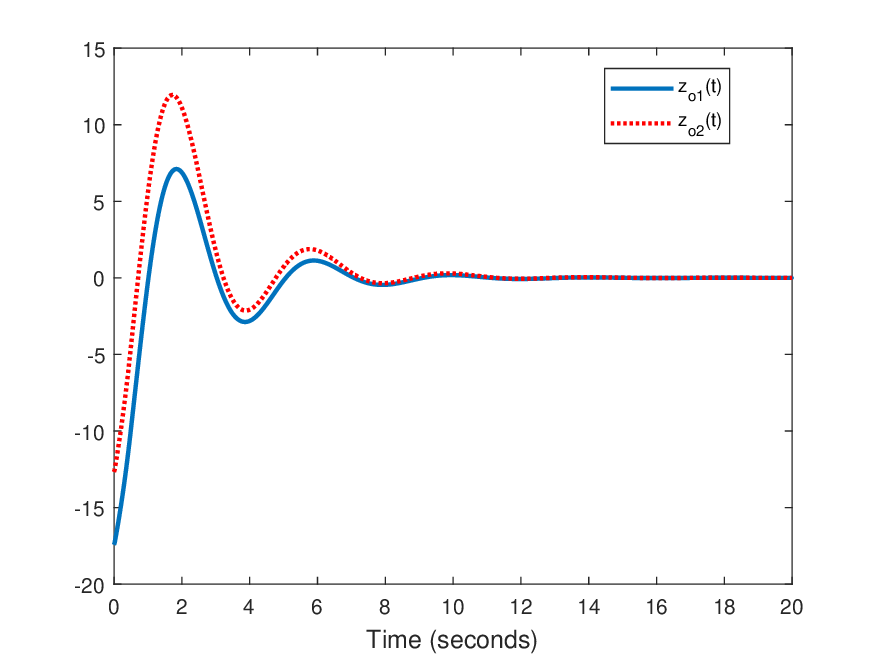}
	\caption{Trajectories of $z_{o1}(t)$ and $z_{o2}(t)$ under the proposed target output control law $u(t-0.5)=Z_{\tau_u}z_o(t-0.5)$ }
	\label{fig1paper6}
\end{figure}

On the other hand, for the case $\tau_x>\tau_u$, say, $\tau_x=1\text{s}$ and $\tau_u=0.5\text{s}$, we can employ the control law (\ref{d5}), i.e.,
$$u(t-0.5) = Z_{\tau_u}z_o(t-0.5) + Z_{\tau_x}z_o(t-1)$$
to stabilize the reduced-order system (\ref{d10}). Thus, for the given $\tau_u=0.5\text{s}$ and $\tau_x=1\text{s}$, the LMI in Lemma 13 \cite{trinhnam26} is feasible for $\lambda=1$, and we obtain
\[Z_{\tau_u}=\begin{pmatrix} 1.3280  & -2.5848\\-1.1172 &1.4174 \end{pmatrix},\]
\[Z_{\tau_x}=\begin{pmatrix} 4.2524  & -5.1863\\
	-2.8542 &   2.6962 \end{pmatrix}.\]
Note that this control law is specifically designed to accommodate a larger time delay $\tau_x$ and ensure a more robustly stable closed-loop system (see \cite{trinhnam26}). Applying this control law to the full-order system (\ref{d1}) yields the following closed-loop system
\begin{align}
	\dot{x}(t)&=Ax(t)+BZ_{\tau_u}F_ox(t-0.5)\nonumber\\&+(A_d+BZ_{\tau_x}F_o)x(t-1).\nonumber \end{align}
As shown in Figure~\ref{fig2paper6}, the trajectories of $z_{o1}(t)$ and $z_{o2}(t)$ converge to zero as $t \rightarrow \infty$, demonstrating a faster convergence rate than those in Figure~\ref{fig1paper6}.
\begin{figure}[!h]
	\centering
	\includegraphics[width=0.9\linewidth]{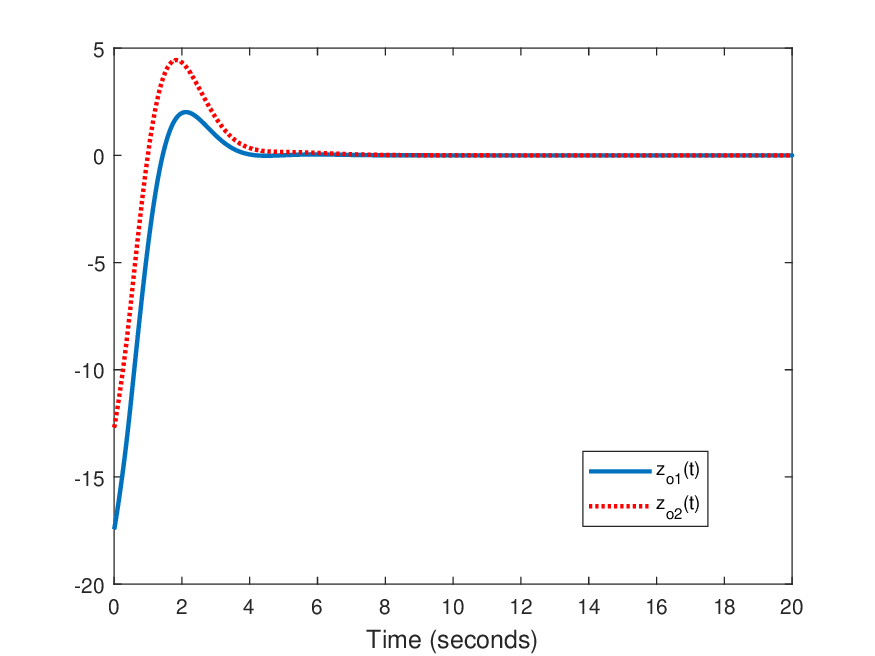}
	\caption{Trajectories of $z_{o1}(t)$ and $z_{o2}(t)$ under the proposed target output control law $u(t-0.5) = Z_{\tau_u}z_o(t-0.5) + Z_{\tau_x}z_o(t-1)$ }
	\label{fig2paper6}
\end{figure}

\textit{Remark 1 (Rank Condition Fulfillment via Augmentation):} When the algebraic rank condition \eqref{d7} does not hold, we can construct an augmented functional matrix $\bar{F}_o$ of the form
$$ \bar{F}_o = \begin{pmatrix} F_o \\ R \end{pmatrix} $$
by choosing an auxiliary matrix $R \in \mathbb{R}^{(q-m) \times n}$, where $m < q \leq n$, such that $\bar{F}_o$ satisfies the rank condition when substituted into \eqref{d7}. Consequently, the controller design can proceed based on an extended $q$-dimensional subsystem of the form
\begin{align}\label{d18}\dot{\bar{z}}_o(t) = \bar{N}_o\bar{z}_o(t) + \bar{N}_{od}\bar{z}_o(t-\tau_x) + \bar{B}_ou(t-\tau_u),\end{align}
where the augmented target output vector is defined as$$\bar{z}_o(t) = \begin{pmatrix} z_o(t) \\ z_{a}(t) \end{pmatrix} = \begin{pmatrix} F_o \\ R \end{pmatrix} x(t) = \bar{F}_ox(t),$$and the corresponding system matrices are given by $\bar{N}_o := \bar{F}_oA\bar{F}^-_o$, $\bar{N}_{od} := \bar{F}_oA_d\bar{F}^-_o$, and $\bar{B}_o := \bar{F}_oB$.

By shifting the control objective to regulating the extended target functional $\bar{z}_o(t)$, this relaxation bypasses the structural rank constraint. In essence, the designed controller now contains more columns in its gain matrices, corresponding to the larger dimension $q$, to overcome an otherwise restrictive algebraic limitation.

Guided by the framework in \cite{Fernando2025}, we construct a matrix $R \in \mathbb{R}^{(q-m) \times n}$ to satisfy the rank condition:
\begin{align}\label{d19}\mathrm{rank}\begin{pmatrix}
	\bar{F}_oA\\\bar{F}_o
\end{pmatrix} =\mathrm{rank}\begin{pmatrix}
	\bar{F}_oA_d\\\bar{F}_o
\end{pmatrix}=\mathrm{rank}(\bar{F}_o).\end{align}

We define the augmented observability matrix as
\begin{equation}
	\label{d20}	\mathcal{O}_{(A,A_d,F_o)}:=	\begin{pmatrix}
		F_o \\ F_oA \\ \vdots \\ F_oA^{n-1} \\
		F_oA_d \\ \vdots \\ F_oA_d^{n-1}
	\end{pmatrix},
\end{equation}
with its rank denoted by
$$q := \mathrm{rank}\left(\mathcal{O}_{(A,A_d,F_o)}\right).$$

From above, we can construct a new full-row-rank matrix, denoted as $\bar{F}_o$, by selecting $q$ linearly independent rows from $\mathcal{O}_{(A,A_d,F_o)}$ such that the rows of $F_o$ are preserved, i.e.,
\begin{equation*}
	\bar{F}_o =
	\begin{pmatrix} F_o \\ R \end{pmatrix}
	\in \mathbb{R}^{q\times n},
\end{equation*}
where $R$ is the $(q - m)\times n$ matrix created by basis rows of $ \{F_oA, F_oA^2, \ldots, F_oA^{n-1},\allowbreak F_oA_d, F_0 A_d^2, \ldots\, F_oA_d^{n-1}\}.$

By the Cayley--Hamilton theorem,
both $F_oA^n$ and $F_oA_d^n$ are linear combinations of
$F_o, F_oA, \ldots, F_oA^{n-1}$ and $F_o, F_oA_d, \ldots, F_oA_d^{n-1}$, respectively.
Hence
\begin{equation*}
	\mathrm{row}(\bar{F}_oA) \subseteq \mathrm{row}(\mathcal{O}_{(A,A_d,F_o)})
	= \mathrm{row}(\bar{F}_o),
\end{equation*}
\begin{equation*}
	\mathrm{row}(\bar{F}_oA_d) \subseteq \mathrm{row}(\mathcal{O}_{(A,A_d,F_o)})
	= \mathrm{row}(\bar{F}_o),
\end{equation*}
and therefore the rank condition (\ref{d19}) is satisfied.

To illustrate Remark 1, let us consider the following example.

\textit{Example 2:}  Consider the system (\ref{d1}) with time delays $\tau_u =0.5\text{s}$, $\tau_x = 1\text{s}$ defined by the following matrices:

$A=\begin{pmatrix} 0  &   1 &    1\\
	0  &   -1  &   2\\
	0  &   -3  &   2 \end{pmatrix}$, $A_d=\begin{pmatrix} 1  &   -1 &    2\\
	0  &   -2  &   1\\
	0  &   -1  &   -2 \end{pmatrix}$,

 $B=\begin{pmatrix}1 &2\\3 &4\\5 &6\end{pmatrix}$,  and $F_o=\begin{pmatrix} 0  & 1  & 1 \end{pmatrix}$.

For the given target matrix $F_o$, it can be verified that the rank condition (\ref{d7}) is not satisfied because
\[\mathrm{rank}\begin{pmatrix}
	F_oA\\F_o
\end{pmatrix} =\mathrm{rank}\begin{pmatrix}
	F_oA_d\\F_o
\end{pmatrix}=2\]
while $\mathrm{rank}(
	F_o)=1$. 
To construct the matrix $R$, we first compute the augmented observability matrix
\begin{equation*}	\mathcal{O}_{(A,A_d,F_o)}=	\begin{pmatrix}
		0 &1&1\\0 &-4 &4\\0 &-8 &0\\0 &-3 &-1\\0 &7 &-1
	\end{pmatrix},
\end{equation*}
yielding $q=2$. Therefore, we can select $R$ as
$$R=\begin{pmatrix}
	0 &-4&4\end{pmatrix},$$
which yields
$$\bar{F}_o=\begin{pmatrix}
	0& 1 &1\\0 &-4&4\end{pmatrix}.$$
It can now be verified that the rank condition (\ref{d19}) is satisfied, yielding an extended second-order subsystem of the form (\ref{d18}) with the following matrices:

$\bar{N}_o=\begin{pmatrix}0 &1\\-4 &1\end{pmatrix}$, $\bar{N}_{od}=\begin{pmatrix}-2 &0.25\\-4 &-2\end{pmatrix}$, and $\bar{B}_o=\begin{pmatrix}8 &10\\8 &8\end{pmatrix}$.

Now, to stabilize the extended target output vector, a simpler delayed control law is initially considered. Specifically, we seek to synchronize the delays in both the input and state vectors (cf. \cite{trinhnn26}). Since $\tau_u<\tau_x$, this synchronization is realized by deliberately introducing a delay of $\tau_x-\tau_u$ into the control input, yielding the following control law:$$u(t-\tau_x)=\bar{Z}_{\tau_x}\bar{z}_o(t-\tau_x),$$where $\bar{Z}_{\tau_x} \in \mathbb{R}^{2 \times 2}$ is a control gain to be determined. However, for a delay of $\tau_x=1\text{s}$, the LMI formulation in Lemma 11 of \cite{trinhnam26} is not feasible. This restriction indicates that the above baseline control law (via the use of Lemma 11 in \cite{trinhnam26}) cannot guarantee the stability of the target output vector, thereby motivating the use of a more generalized delayed control law of the form (\ref{d5}):
\[u(t-\tau_u) = \bar{Z}_{\tau_u}\bar{z}_o(t-\tau_u) + \bar{Z}_{\tau_x}\bar{z}_o(t-\tau_x)\]
where $\tau_u=0.5\text{s}$ and $\tau_x=1\text{s}$. The control gains $\bar{Z}_{\tau_u}, \bar{Z}_{\tau_x} \in \mathbb{R}^{2 \times 2}$ are then designed to ensure that the following closed-loop augmented target output system is asymptotically stable:
\begin{align*}\dot{\bar{z}}_o(t) &= \bar{N}_o\bar{z}_o(t) +\bar{B}_o\bar{Z}_{\tau_u}\bar{z}_o(t-\tau_u)\\&+ (\bar{N}_{od}+\bar{B}_o\bar{Z}_{\tau_x})\bar{z}_o(t-\tau_x) .\end{align*}

For the given $\tau_u=0.5\text{s}$ and $\tau_x=1\text{s}$, the LMI in Lemma 13 \cite{trinhnam26} is feasible for $\lambda=1$, and we obtain
\[\bar{Z}_{\tau_u}=\begin{pmatrix} 0.6958 &-0.1612\\-0.5708 &0.0536 \end{pmatrix},\]
\[\bar{Z}_{\tau_x}=\begin{pmatrix} 1.6238 &1.2896\\-1.1135 &-1.0542 \end{pmatrix}.\]

The designed control law is subsequently applied to the full-order system. As illustrated in Figure~\ref{fig3paper6}, the trajectories of $z_{o}(t)$ and $z_a(t)$ asymptotically converge to zero, successfully validating Remark 1.
\begin{figure}[!h]
	\centering
	\includegraphics[width=0.9\linewidth]{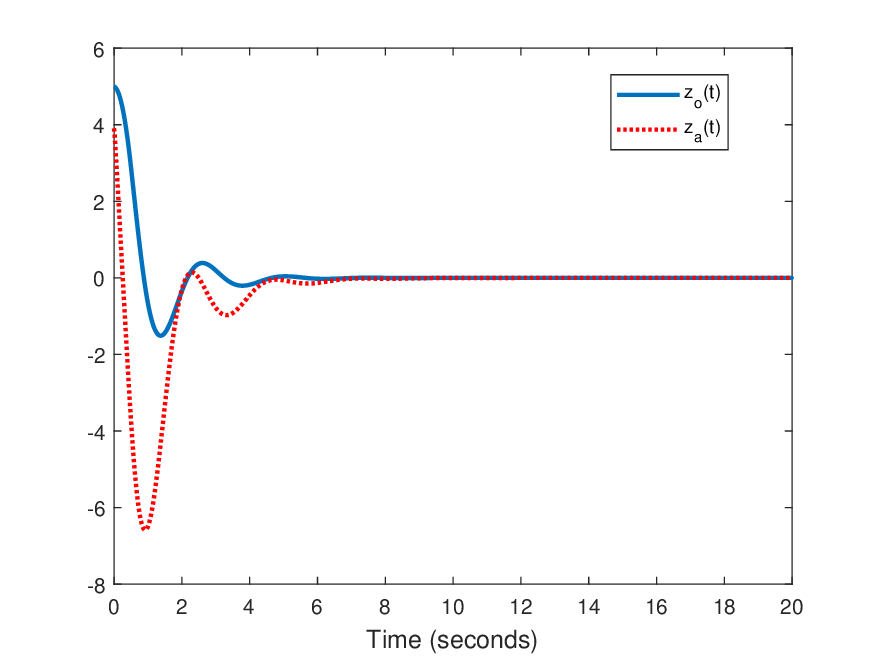}
	\caption{Trajectories of $z_{o}(t)$ and $z_{a}(t)$ under the proposed target output control law $u(t-0.5) = \bar{Z}_{\tau_u}\bar{z}_o(t-0.5) + \bar{Z}_{\tau_x}\bar{z}_o(t-1)$ }
	\label{fig3paper6}
\end{figure}

\textit{Remark 2:} To handle the case where $\tau_u > \tau_x$, the proposed control law can take either of the following forms:
\begin{align}
	u(t) &= Z_{\tau_u}z_o(t),\\ 
	u(t) &= Z_{\tau_u}z_o(t) + Z_hz_o(t - h), \ h>0.
\end{align}
Consequently, under the input delay $\tau_u$, the actual control signals entering the system become:
\begin{align}
	\label{control1}
	u(t-\tau_u) &= Z_{\tau_u}z_o(t-\tau_u),\\
	\label{control2}
	u(t-\tau_u) &= Z_{\tau_u}z_o(t-\tau_u)+ Z_hz_o(t-\alpha), \end{align}
where $\alpha=\tau_u+h$, and $Z_{\tau_u}, Z_h \in \mathbb{R}^{r \times m}$ are the control gains to be designed such that the target output vector $z_o(t)\to \bf 0$ asymptotically.

Substituting (\ref{control1}) and (\ref{control2}) into (\ref{d10}) yields the following closed-loop time-delay systems:
\begin{align}\label{d25}\dot{z}_o(t)&= N_oz_o(t) +N_{od}z_o(t-\tau_x)+B_oZ_{\tau_u}z_o(t-\tau_u),\\
\label{d26}\dot{z}_o(t)&= N_oz_o(t) +N_{od}z_o(t-\tau_x)+B_oZ_{\tau_u}z_o(t-\tau_u)\nonumber\\&\quad +B_oZ_hz_o(t-\alpha).\end{align}
Thus, the objective is to determine the control gains $Z_{\tau_u}$ and $Z_h$ to guarantee the asymptotic stability of these respective closed-loop systems.

Note that the stabilization of \eqref{d25} is more restrictive than that of \eqref{d11}. Nevertheless, given the delays $\tau_u$ and $\tau_x$, Lemma 13 in \cite{trinhnam26} remains applicable for determining the gain matrix $Z_{\tau_u}$. In contrast, the stabilization of \eqref{d26} involves three distinct mismatched time delays satisfying $\tau_x<\tau_u<\alpha$. For this scenario, the control gains $Z_{\tau_u}$ and $Z_h$ can be determined via Lemma 14 in \cite{trinhnam26} subject to the feasibility of its formulated LMI. As illustrated in Example A8 of \cite{trinhnam26}, stabilizing \eqref{d26} is less restrictive than stabilizing \eqref{d25}, highlighting a clear engineering trade-off.

\section{Observer-Based Target Output Control Using Lower-Order Time-Delay Subsystems}
In Section II, we demonstrated that the target output vector $z_o(t)$ can be effectively stabilized using a delayed control law, designed based on either an $m$- or $q$-dimensional subsystem of the form (\ref{d10}) or (\ref{d18}). In practice, however, neither $z_o(t)$ nor its delayed state $z_o(t-\tau_u)$ is directly available for feedback. Furthermore, the measured output vector is subject to the time delay modeled in (\ref{d2}). To overcome these limitations, the designed controller is implemented via a functional observer.

According to \cite{trinh5}, the delayed control law (\ref{d5}) is estimated by estimating the following functional
\begin{align}
	\label{d27}
	z(t)&:= Z_{\tau_u}z_o(t-\tau_u) + Z_{\tau_x}z_o(t-\tau_x)\nonumber\\&=Z_{\tau_u}F_ox(t-\tau_u) + Z_{\tau_x}F_ox(t-\tau_x) \end{align}
based on the full-order system (\ref{d1})-(\ref{d2}). As recently discussed in \cite{trinh5} for the case without state time-delays (i.e., $A_d=\mathbf{0}$), the estimation of (\ref{d27}) can be achieved using a dual-observer architecture. This dual-observer scheme can be generalized to the class of time-delay systems (\ref{d1})-(\ref{d2}) with mismatched time delays. In the following, we restrict our attention to the case where $\tau_u < \tau_y < \tau_x$. The analysis of the alternative ordering, $\tau_y > \tau_x > \tau_u$, follows a similar trajectory and may be pursued using foundational insights available in recent literature \cite{trinh5, trinhnam26, trinhnn26,  trinh2026existence, trinh1, trinh2}).

Without loss of generality, we assume that $Z_{\tau_u}F_o$ and $Z_{\tau_y}F_o$ have full row rank. We then rewrite \eqref{d27} as follows:
\begin{align}
	\label{c6}
	z(t)=z_1(t)+z_2(t),
\end{align}
where
\begin{align*}
	z_1(t) &:=Z_{\tau_u}F_ox(t-\tau_u)=F_ux(t-\tau_u),\\
	 \quad z_2(t) &:=Z_{\tau_x}F_ox(t-\tau_x)= F_xx(t-\tau_x).
\end{align*}
As in \cite{trinh5}, we employ a dual-observer scheme where Observer 1 estimates $z_1(t)$ and Observer 2 estimates $z_2(t)$, yields the overall estimate:
\begin{align*}
	\hat{z}(t)=\hat{z}_1(t)+\hat{z}_2(t).
\end{align*}

The dual-observer architecture enables stable control under severe latency by splitting the estimation task into two reduced-order, parallel estimators, eliminating the need for full-state reconstruction. Because the second functional may partially align with the spaces spanned by the output matrix $C$ and the matrix $F_u$, matrix decomposition techniques can be applied. This allows Observer 2 to feature a structurally simpler design, thereby significantly reducing computational complexity.

Following the framework in \cite{trinh2}, we consider the following novel observer to asymptotically estimate $z_1(t)$. This structure allows us to account for mismatched time delays $\tau_u$, $\tau_y$ and $\tau_x$, while achieving a more robust and less conservative stability for the error time-delay system across larger ranges of $\tau_u$ and $\tau_x$.
\begin{align}
	&\hat{z}_1(t)=w_1(t)+M_1y(t)+M_{1\beta}y(t-\beta),\label{d29}\\
	&\dot{w}_1(t)=N_1w_1(t)+N_{1\tau}w_1(t-\tau)+N_{1\tau_x}w_1(t-\tau_x)\nonumber\\&+G_1y(t)+G_2y(t-\beta)+G_3y(t-\tau)+G_4y(t-\tau_x)\nonumber\\&+G_5y(t-\beta -\tau_x)+Ju(t-2\tau_u)+J_1u(t-\tau_u-\tau_y)\nonumber\\&+J_2u(t-\tau_u-\tau_y-\beta),\label{d30}
\end{align}
where $\tau=\tau_y-\tau_u$, $\beta=\tau_u+\tau_x-\tau_y$, and the initial condition is given by $w_1(\theta) = \rho(\theta)$ for $\theta \in [-\tau_x, 0]$, and $\hat{z}_1(t)$ is the estimate of $z_1(t)$. Matrices $M_1$, $M_{1\beta}$, $N_1$, $N_{1\tau}$, $N_{1\tau_x}$, $G_i$ ($i=1,2,\cdots,5$), $J$, $J_1$ $J_2$ are to be determined such that $\hat{z}(t)$ converges asymptotically to $z(t)$. 

We next derive the existence conditions for the observer (\ref{d29})-(\ref{d30}). To this end, we define the estimation error vector as $e(t)=\hat{z}_1(t)-z_1(t)$. Given $\tau$ and $\beta$ as defined above, the following time-delay relationships hold:
\begin{align*}x(t-\tau_y)&=x(t-\tau-\tau_u),\\
	x(t-\tau_u-\tau_x)&=x(t-\beta-\tau_y),\\
	x(t-\tau_y-\tau_x)&=x(t-\tau_y-\beta-\tau).
	\end{align*}
Consequently, the resulting error dynamics can be derived and expressed as
\begin{align}
	\label{d31}
	&\dot{e}(t)=\dot{w}_1(t)+M_1\dot{y}(t)+M_{1\beta}\dot{y}(t-\beta)-F_u\dot{x}(t-\tau_u)\nonumber\\  &=N_1e(t)+N_{1\tau}e(t-\tau)+N_{1\tau_x}e(t-\tau_x)+\mathcal{C}_{1}u(t-2\tau_u)\nonumber\\  &+\mathcal{C}_{2}u(t-\tau_u-\tau_y) +\mathcal{C}_{3}u(t-\tau_u-\tau_y-\beta)+\mathcal{C}_{4}x(t-\tau_u) \nonumber\\  &+\mathcal{C}_{5}x(t-\tau_y)+\mathcal{C}_{6}x(t-\tau_u-\tau_x)+\mathcal{C}_{7}x(t-\tau_y-\tau_x)\nonumber\\ &+\mathcal{C}_{8}x(t-\tau_y-\beta-\tau_x)+\mathcal{C}_{9}x(t-\tau-\tau_y),
\end{align}
where\\ 
\\
$\mathcal{C}_1=J-F_uB$, $\mathcal{C}_2=J_1+M_1CB$, $\mathcal{C}_3=J_2+M_{1\beta}CB$,  $\mathcal{C}_4=N_1F_u-F_uA$, 
$\mathcal{C}_5=N_{1\tau}F_u+\bar{G}_1C+M_1CA$,  $\mathcal{C}_6=N_{1\tau_x}F_u+\bar{G}_2C+M_{1\beta}CA-F_uA_d$, $\mathcal{C}_7=\bar{G}_4C+M_1CA_d$, $\mathcal{C}_{8}=\bar{G}_5C+M_{1\beta}CA_d$, $\mathcal{C}_{9}=\bar{G}_3C$,\\
\\
$\bar{G}_1:=G_1-N_1M_1$,  $\bar{G}_2:=G_2-N_1M_{1\beta}$,  $\bar{G}_3:=G_3-N_{1\tau}M_1$,  $\bar{G}_4:=G_4-N_{1\tau_x}M_1-N_{1\tau}M_{1\beta}$, $\bar{G}_5:=G_5-N_{1\tau_x}M_{1\beta}$ \ and \
$\mathcal C
:=
\begin{pmatrix}
	\mathcal C_1 &
	\mathcal C_2 &
	& \cdots &
	\mathcal C_{8} &
	\mathcal C_{9}
\end{pmatrix}.$

The following theorem provides conditions for the existence of observer (\ref{d29})-(\ref{d30}).

\begin{thm}\label{thm:1p7}
Observer (\ref{d29})-(\ref{d30}) provides asymptotic estimation of the functional $z_1(t)=F_ux(t-\tau_u)$ and yields
estimation error dynamics that are decoupled from the plant state $x(\cdot)$ and the input $u(\cdot)$
if $\mathcal C=\bf 0$ and the following delay-dependent error dynamics
\begin{align}
	\label{d32}&\dot{e}(t)=N_1e(t)+N_{1\tau}e(t-\tau)+N_{1\tau_x}e(t-\tau_x)
\end{align}
is asymptotically stable. In this case, the estimation error satisfies
\[
e(t)=\hat z_1(t)-z_1(t)\to {\bf 0} \quad \text{as} \quad t\to\infty
\]
for all admissible initial conditions and inputs $u(\cdot)$.
\end{thm}

\begin{proof}	If $\mathcal C=\bf 0$, then \eqref{d31} reduces to \eqref{d32}, so the error dynamics are
	decoupled from $x(\cdot)$ and $u(\cdot)$. If, in addition, \eqref{d32} is asymptotically stable, then $e(t)\to \bf 0$ as $t\to\infty$ for all admissible initial conditions and inputs $u(\cdot)$. This completes the proof.
\end{proof}

\textit{Remark 3:} By explicitly incorporating the net time-delay $\tau = \tau_y - \tau_u$, observer (\ref{d29})-(\ref{d30}) yields a less conservative stability condition than the error system in \cite{trinh2} (see equation (9) therein). It also reduces structural complexity compared to observer (6)-(7) in \cite{trinh2}, representing a clear improvement. For brevity, we omit the corresponding observer structures and existence conditions for the cases $\tau_y \leq \tau_u$ or $\tau_y = \tau_x$, as well as the derivation of the observer gains to satisfy Theorem 1, all of which follow a design procedure similar to that in \cite{trinh2}.

With the observer for $z_1(t)$ established, we now design a second observer to estimate the functional $z_2(t)=F_xx(t-\tau_x)$. Because $\tau_u<\tau_y < \tau_x$ and $F_x$ may partially reside within the subspaces spanned by the matrices $C$ and $F_u$, this second observer benefits from less conservative existence conditions and a simpler structure, significantly reducing computational overhead. These advantages, originally highlighted by Trinh \cite{trinh5}, motivate the dual-observer architecture proposed in this work to estimate $z(t)$ (note that an alternative approach for direct estimation of $z(t)$ was recently proposed in \cite{trinhvan26}).

Further details on this architecture are available in \cite{trinh5}, which explores the case where a subset of the rows in $F_x$ lies within the row space of $\begin{pmatrix} F_u \\ C \end{pmatrix}$. Under this geometric condition, $F_x$ can be decomposed using a transformation matrix $K$ and a residual matrix $\bar{F}_x$. The rows of $\bar{F}_x$ are orthogonal to the row spaces of $C$ and $F_u$ via orthogonal projection, ensuring linear independence. The full derivation of this decomposition is omitted here for brevity. Consequently, the designer only needs to construct a functional observer for the remaining, lower-dimensional sub-functional:
$$\bar{z}_2(t)=\bar{F}_xx(t-\tau_x).$$
This presents a significant architectural advantage, as it reduces the observer's computational complexity while maintaining full estimation integrity.

\textit{Example 3:}  Let us reconsider Example 2 under the condition that the output delay is $\tau_y = 0.8\text{s}$ and the output matrix $C$ is specified as
\[C=\begin{pmatrix} 0  &  0 &    1 \end{pmatrix}.\]
In Example 2, with the given delays $\tau_u=0.5\text{s}$ and $\tau_x=1\text{s}$, we designed a control law of the form
 \[u(t-0.5) = \bar{Z}_{\tau_u}\bar{z}_o(t-0.5) + \bar{Z}_{\tau_x}\bar{z}_o(t-1),\]
where
\[\bar{Z}_{\tau_u}=\begin{pmatrix} 0.6958 &-0.1612\\-0.5708 &0.0536 \end{pmatrix},\]
\[\bar{Z}_{\tau_x}=\begin{pmatrix} 1.6238 &1.2896\\-1.1135 &-1.0542 \end{pmatrix},\] to asymptotically stabilizing the augmented target output
 \begin{align*}\bar{z}_o(t)=\bar{F}_ox(t)=\begin{pmatrix}
 		0& 1 &1\\0 &-4&4\end{pmatrix}x(t) .\end{align*}

Because the only measured output is the state variable $x_3(t)$ delayed by $0.8\text{s}$, neither $\bar{z}_o(t)$ nor $\bar{z}_o(t - 0.5)$ is directly available for feedback control. Therefore, following the framework developed above, we now employ a dual-observer scheme to reconstruct the designed control law.

The first observer is constructed to estimate the following functional
\[z_1(t)=F_ux(t-0.5),\]
where $F_u=\bar{Z}_{\tau_u}\bar{F}_o=\begin{pmatrix}
	0 &1.3406 &0.0511\\0 &-0.7852 &-0.3563\end{pmatrix}$.

As mentioned in Remark 3, the observer gains required to satisfy Theorem 1 can be computed by following a design procedure similar to the one in \cite{trinh2}, yielding the following observer:
\begin{align*}
	\hat{z}_1(t)& = w_1(t),\\
	\dot{w}_1(t)&=\begin{pmatrix}  -6.2106 &-8.7008\\5.6066 &   7.2106 \end{pmatrix}w_1(t)\\&+\begin{pmatrix}6.0612 &10.3481\\-5.0014 &-8.5388\end{pmatrix}w_1(t-0.3)\\&+\begin{pmatrix}-4.3510 &-3.9487\\3.0239 &2.7088\end{pmatrix}w_1(t-1) \nonumber\\&+\begin{pmatrix} 3.3778\\-2.7872 \end{pmatrix}y(t)+\begin{pmatrix} 0.0536\\0.7383\end{pmatrix}y(t-0.7)\nonumber\\&+\begin{pmatrix} 4.2770 &5.6686\\-4.1372 &-5.2788 \end{pmatrix}u(t-1), \nonumber
\end{align*}
where
\[w_1(t)=\begin{pmatrix} w_{11}(t)\\
	w_{12}(t) \end{pmatrix}.\]
Finally, for the second functional $z_2(t)$, the fact that $F_x$ lies within the row space of $F_u$ (and $C$) allows us to directly reconstruct $z_2(t)$ from $\hat{z}_1(t)$ as follows: 
$$\hat{z}_2(t)=K_2\hat{z}_1(t-\gamma),$$
where
\[\gamma = \tau_x-\tau_u=0.5\text{s}, \ K_2=F_xF_u^-=\begin{pmatrix} -15.0485 &-21.1903\\12.0913 &16.6912\end{pmatrix}.\]

Integrating the designed controller with the dual-observer structure ensures the asymptotic stability of the target output $z_o(t)$. This is confirmed by the trajectories in Figure \ref{fig4paper5}, which show $z_o(t)$ converging asymptotically to zero.
\begin{figure}[!h]
	\centering
	\includegraphics[width=0.9\linewidth]{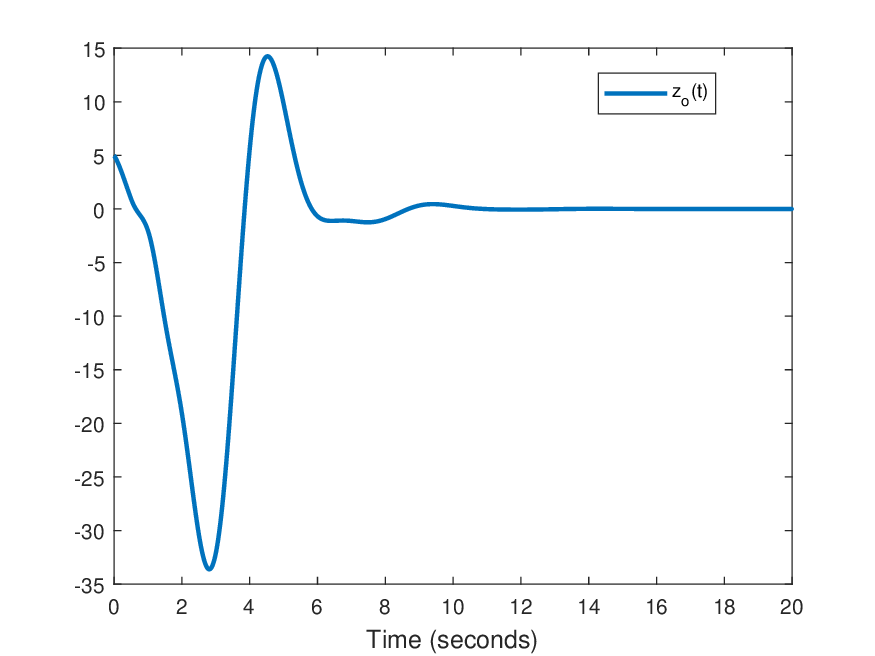}
	\caption{Trajectories of $z_o(t)$: Observer-based control}
	\label{fig4paper5}
\end{figure}

The remainder of this paper explores an alternative approach based on \cite{trinhnn26}. In this method, the delayed control law (\ref{d5}) is estimated using the reduced-order time-delay subsystem (\ref{d10}), whose output forms a subset of the overall output vector. For a more detailed discussion of this approach, see Section II-D of \cite{trinhnn26}.

\textit{Example 4:} Consider again the system in Example 1 with the time delays given by $\tau_u = 0.5\text{s}$, $\tau_x = 1\text{s}$, and $\tau_y = 0.8\text{s}$. The output matrix $C$ is chosen to match $F_o$, defined as
\[C=F_o=\begin{pmatrix}0.5 &1 &-2 &0.5 &2.5\\0.75 &1 &-2 &0.25 &2.25 \end{pmatrix}.\]

In Example 1, with the given delays  $\tau_u=0.5\text{s}$ and $\tau_x=1\text{s}$, we designed a control law of the form
$$u(t-0.5) = Z_{\tau_u}z_o(t-0.5) + Z_{\tau_x}z_o(t-1)$$
where
\[Z_{\tau_u}=\begin{pmatrix} 1.3280  & -2.5848\\-1.1172 &1.4174 \end{pmatrix},\]
\[Z_{\tau_x}=\begin{pmatrix} 4.2524  & -5.1863\\
	-2.8542 &   2.6962 \end{pmatrix}\]
to asymptotically stabilizing the target $z_o(t)=F_ox(t)$.

Now, this example highlights an interesting scenario. Even though the output matrix $C$ is identical to $F_o$, the output measurement vector is delayed by $\tau_y$. Because $\tau_y > \tau_u$, the delayed target variable $z_o(t-0.5)$ cannot be directly inferred from the measurement $y(t)$, which is itself delayed as $y(t) = z_o(t-0.8)$. Consequently, an observer (or a time-delay compensator) is strictly required to implement the control law in this practical scenario.

Since $C=F_o$, we can now express the output vector of this example in the following form
\[y(t)=F_ox(t-0.8)=C_oz_o(t-0.8),\]
where $C_o=I_2$.
Thus, we obtain the following reduced-order subsystem:
\begin{align*}
	\dot{z}_o(t)&=N_oz_o(t) + N_{od}z_o(t-1) +B_ou(t-0.5), \\
	y(t)&=z_o(t-0.8),
\end{align*}
where $N_o=\begin{pmatrix}0 &1\\-0.5 &1.5\end{pmatrix}$,  $N_{od}=\begin{pmatrix}0 &2.5\\-1.25 &3.75\end{pmatrix}$, and $B_o=\begin{pmatrix}2 &3\\2 &2.5\end{pmatrix}$.

The objective now is to estimate the functional$$z_1(t)=Z_{\tau_u}z_o(t-0.5)$$utilizing the delayed output measurement vector $y(t)=z_o(t-0.8)$. To this end, by following the framework presented in \cite{trinh1} (see Remark 5 therein), we construct the following observer:
\begin{align*}
	\dot{\hat{z}}_1(t)&=N_1\hat{z}_1(t)+N_{1\tau}\hat{z}_1(t-\tau)+N_{1\tau_x}\hat{z}_1(t-\tau_x)\nonumber\\&+G_1y(t)+G_2y(t-\beta)+Ju(t-2\tau_u).
\end{align*}
The detailed derivation of the observer gains is omitted here for brevity (the reader may refer to Remark 5 and Corollary 1 in \cite{trinh1}). Accordingly, we obtain the following observer to asymptotically estimate $z_1(t)$:
\begin{align*}
	\dot{\hat{z}}_1(t)&=\begin{pmatrix}1.0106 &0.0445\\-0.1220 &0.4894 \end{pmatrix}\hat{z}_1(t)\\&+\begin{pmatrix}-1.3095 &0.0113\\0.1013 &-0.8099\end{pmatrix}\hat{z}_1(t-0.3)\\&+\begin{pmatrix}-0.1753 &-0.0081\\-0.0009 &-0.2151\end{pmatrix}\hat{z}_1(t-1)\\&+\begin{pmatrix} 1.7517 &-3.4008\\-1.0394 &1.4098 \end{pmatrix}y(t)\\&+\begin{pmatrix}3.4549 &-6.8147\\
-2.0108 &2.8248 \end{pmatrix}y(t)\\&+\begin{pmatrix}-2.5136 &-2.4779\\0.6005 &0.192 \end{pmatrix}u(t-1). \nonumber
\end{align*}

Regarding the second functional $z_2(t)=Z_{\tau_x}z_o(t-1)$, the fact that $y(t)=z_o(t-0.8)$ allows us to obtain $z_2(t)$ from $y(t)$ as follows: 
$$z_2(t)=Z_{\tau_x}y(t-0.2),$$

Finally, integrating the designed controller with the dual-observer structure
\[\hat{z}(t)=\hat{z}_1(t)+Z_{\tau_x}y(t-0.2)\]
ensures the asymptotic stability of the target output $z_o(t)$. This is confirmed by the trajectories in Figure \ref{fig5paper5}, which show $z_o(t)$ converging asymptotically to zero.
\begin{figure}[!h]
	\centering
	\includegraphics[width=0.9\linewidth]{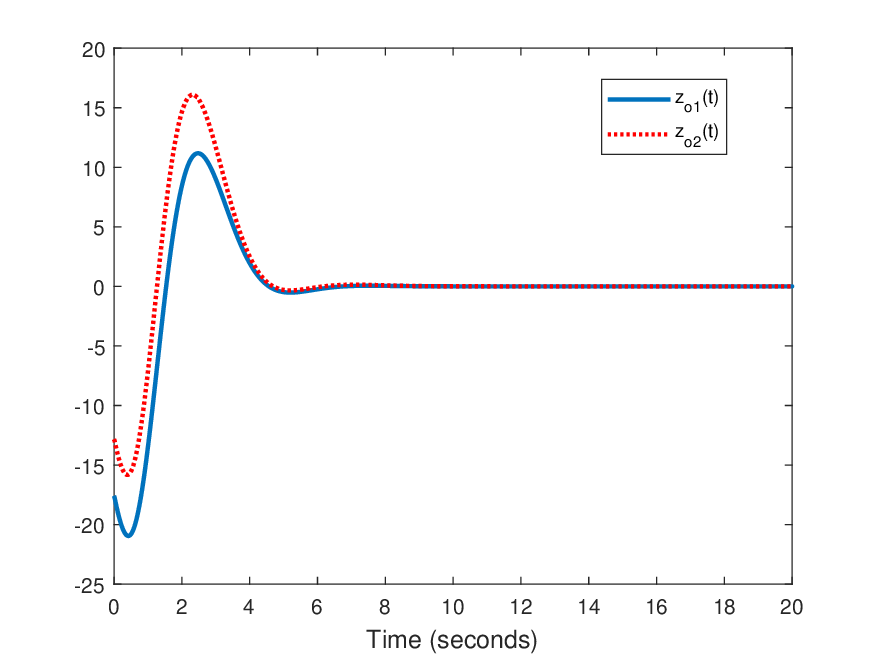}
	\caption{Trajectories of $z_{o_1}(t)$ and $z_{o_2}(t)$: Observer-based control}
	\label{fig5paper5}
\end{figure}

\section{Conclusion} 
Together, this paper and the companion preprints \cite{trinh5, trinhnam26, trinhnn26, trinh2026existence, trinh1, trinh2} constitute a multi-part study detailing distinct advancements within time-delay systems theory. Specifically, this body of work establishes systematic synthesis procedures for delayed functional observers and time-delay compensators, focusing on the mitigation of simultaneous, mismatched time delays across system states, control inputs, and output measurements.

Promising directions for future research include: (i) addressing intermittent measurement scenarios where delayed outputs are subject to temporary outages, necessitating a switched-observer approach to guarantee stability during data blackouts; (ii) extending these low-order delayed functional observers to handle time-varying mismatched delays $\tau_u(t)$, $\tau_y(t)$, and $\tau_x(t)$; (iii) adapting the compensator architecture to broader classes of nonlinear or uncertain systems; (iv) integrating Unknown Input Observer (UIO) principles into this multi-part framework to robustly handle external disturbances or faults; and (v) formulating an event-triggered delayed functional observer scheme based on this work to minimize network communication overhead. Further results along these lines will be reported in future works. Accompanying video demonstrations for these and upcoming preprints will be made available on the author's control Vlog channel.

\end{document}